\documentclass[11pt]{article}

\usepackage{amsthm,amsmath,amssymb}
\usepackage{graphicx,color,epsfig}
\usepackage[margin=1in]{geometry}

\newtheorem{thm}{Theorem}[section]

\newtheorem{lem}[thm]{Lemma}
\newtheorem{prop}[thm]{Proposition}

\newtheorem{rmk}[thm]{Remark}
\newtheorem{ex}[thm]{Example}

\newcommand{\N}{\mathbb{N}}
\newcommand{\R}{\mathbb{R}}
\newcommand{\Z}{\mathbb{Z}}
	
\DeclareMathAlphabet{\mathpzc}{OT1}{pzc}{m}{it}

\begin{document}

\title{Asymptotic periodicity in networks of degrade-and-fire oscillators}
\author{Alex Blumenthal$^1$ and Bastien Fernandez$^2$\footnote{On leave from Centre de Physique Th\'{e}orique, CNRS - Aix-Marseille Universit\'e - Universit\'e de Toulon, Campus de Luminy, 13288 Marseille CEDEX 09 France}}
\maketitle

\begin{center}
$^1$ Courant Institute of Mathematical Sciences\\
New York University\\
New York, NY 10012, USA

\vskip0.5cm

$^2$ Laboratoire de Probabilit\'es et Mod\`eles Al\'eatoires\\
CNRS - Universit\'e Paris 7 Denis Diderot\\
75205 Paris CEDEX 13 France\\
\end{center}

\begin{abstract}
Networks of coupled degrade-and-fire (DF) oscillators are simple dynamical models of assemblies of interacting self-repressing genes. For mean-field interactions, which most mathematical studies have assumed so far, every trajectory must approach a periodic orbit. Moreover, asymptotic cluster distributions can be computed explicitly in terms of coupling intensity, and a massive collection of distributions collapses when this intensity passes a threshold.
Here, we show that most of these dynamical features persist for an arbitrary coupling topology. In particular, we prove that, in any system of DF oscillators for which in and out coupling weights balance, trajectories with reasonable firing sequences must be asymptotically periodic, and periodic orbits are uniquely determined by their firing sequence. 
In addition to these structural results, illustrative examples are presented, for which the dynamics can be entirely described. 
\end{abstract}

\section{Introduction}
To predict the long-term behavior in networks of interacting units is a predominant challenge in nonlinear science, with applications in many disciplines, from physics to biology and to the social sciences, to cite a few examples \cite{S01}. In particular, a recurrent question is to characterize collective properties such as synchronisation and predictability in terms of the network topology and interaction strengths \cite{AD-GKMZ08,BCFV02}. While this problem has received considerable attention from theoreticians, mathematically rigorous descriptions of (global) nonlinear behavior are scarce, and only address limited circumstances, {\sl e.g.}\ weak-coupling regimes \cite{M96} and assemblies of pulse coupled oscillators with excitatory coupling \cite{B95,MS90,SU00}. Hence, the theory remains largely incomplete and network phenomenology still lacks a comprehensive rigorous footing.

In the last years, a model for the population dynamics of simple gene oscillators was introduced \cite{FT11}, inspired from a series of experiments on colonies of synthetic genetic circuits \cite{DM-PTH10,M-PDSTH11}, and resulting from the simplification of more standard delay-differential equation models \cite{MBHT09}. In a few words (see section 2 below for more details), it consists of a collection of pulse coupled oscillators with inhibitory coupling \cite{EPG95}, and is reminiscent of the well-known integrate-and-fire model in neuroscience \cite{B06}; however, the phenomenologies of each are distinct.

In the case of mean field coupling, a mathematically rigorous global description of the dynamics, notably its clustering and asymptotic properties, was achieved for every parameter value and for arbitrary numbers of oscillators (and also for the continuum approximation) \cite{BF14,FT11}. Analogous features were also described for trajectories issued from typical random initial conditions \cite{FT13}. In addition, a recent study expanded the analysis to a more elaborate model that involves a global activator field in the dynamics \cite{MHT14}.
Motivated by including more realistic features in the mathematical analysis, the current paper aims to extend previous (deterministic) results to arbitrary coupling topologies on populations of arbitrary size. 

A typical property of degrade-and-fire models is their firing process (accompanied with gene level resetting) that occurs when the repressor field becomes (locally) negligible and can no longer prevent gene expression. (Instantaneous resets are used here as a naive representation of massive gene expression during a tiny interval of time.) 
In the case of mean field coupling, after cell $i$ has fired, every other cell (not simultaneously firing with $i$) must fire once, before $i$ fires again. The ordering in which cells fire does not change from cycle to cycle (unless cells synchronize and begin firing together). More importantly, this periodic and exhaustive cycle of firings imposes asymptotic periodicity on the trajectories themselves. 

While periodic exhaustive firing may not always hold for an arbitrary coupling topology (examples will be provided below), the main result of this paper (Theorem \ref{ASYMPTHM} below)  states that, when this is the case, the trajectory must asymptotically approach a periodic configuration, provided that all cells are path-wise connected through coupling and in and out weights balance at every node. Even though this conclusion does not {\sl a priori} cover all trajectories of degrade-and-fire systems, it proves that a sufficiently regular firing behavior implies a regular asymptotic behavior of the trajectories themselves, under a mild restriction on the coupling structure. Together with the analysis of orbits with symmetric components, this result paves the way to a comprehensive understanding of the functioning of arbitrary systems of coupled DF oscillators.

The paper is organised as follows. The DF model of $N$-oscillators is defined in section 2 and global well-posedness of the dynamics is proved. In section 3, we study properties of the firing events, and use these features to introduce non-degenerate trajectories with exhaustive firing sequences; such trajectories are at the center of attention in the rest of the paper. In section 4, we prove that there can be at most one periodic orbit associated with each such sequence, and provide examples of existence and non-existence, in the case of nearest neighbor coupling. Section 5 contains Theorem \ref{ASYMPTHM} and its proof, while the paper is completed, in section 6, with a study of the full dynamics for $N=2$ and $N=3$ cells (assuming some coupling symmetry in the latter case).    

\section{The degrade-and-fire dynamics}
We consider the degrade-and-fire (DF) dynamics of single self-repressor genes in a colony of cells, driven by intercellular coupling \cite{FT11}. In this context, cells are indexed by $\{1,\cdots,N\}$ (where $N\in\N$) and gene expression levels at time $t\in\R^+$ are represented by the vector $x(t)=(x_i(t))_{i=1}^N\in [0,1]^N$.
Intercellular coupling of genes in this population is mitigated by a repressor field $Wx=(Wx_i)_{i=1}^N$, defined as the action of the linear operator $W$ on the vector $x$,  
\[
W x_i = \sum_{j=1}^N w_{ij} x_j,\ \forall i\in \{1,\cdots,N\},
\]
where the symbol $W=(w_{ij})_{i,j=1}^N$ also denotes a stochastic non-negative matrix. The dynamics depends as well on a threshold parameter $\eta \in (0,1)$, which is assumed to be small. Finally, we impose that the matrix diagonal terms satisfy $w_{ii} > \eta$ for all $i$. (NB: Ref.\ \cite{FT11} assumed {\bf mean field} coupling, {\sl viz.}\ $Wx_i=(1-\epsilon)x_i+\frac{\epsilon}{N}\sum_{j=1}^Nx_j$ for all $i$.) Here, we consider for now any coupling satisfying the condition $w_{ii} > \eta$ for all $i \in \{1,\cdots,N\}$, and later impose additional constraints when asymptotic periodicity is investigated.

With these definitions in place, the DF time evolution of gene expression levels is given by the following differential equation, inspired by the delay-differential equation model in \cite{MBHT09}:
\begin{align}\label{eq:diffEQ}
\begin{split}
\begin{array}{l l}
\dot{x}_i(t) = - \text{Sgn}(x_i(t)) & \text{if} \quad W x_i(t) > \eta, \\
\bigg\{ \begin{array}{l}x_i(t) = x_i(t-0) \\ x_i(t + 0) = 1 \end{array} & \text{if} \quad W x_i(t) \leq \eta.
\end{array} 
\end{split}
\end{align}
In other words, the dynamics in cell $i$ consists of two phases, depending on the repressor field $Wx_i(t)$.
\begin{itemize}
\item  When $W x_i(t) > \eta$, the expression level $x_i(t)$ {\bf degrades} at constant speed $-1$, unless it has reached zero (in which case, it remains at zero). In this phase, if we also have $x_j(t) > 0$ for all cells $j$ such that $w_{ij}>0$ (called influencing cells), the repressor level $W x_i(t)$ also decreases with speed 1. We may eventually have $W x_i(t)\leq \eta$, depending on expression level behaviors in influencing cells. 
\item When $W x_i(t) \leq \eta$, a {\bf firing} takes place and resets the expression level to the value 1. The assumption $w_{ii} > \eta$ ensures that $W x_i(t+0) > \eta$ for the repressor field in cell $i$ after resetting. Hence, after every firing, the reset genes return to the degrade phase for a positive-length time interval. 
\end{itemize}
Accordingly, the behavior in each cell consists of an eternal succession of degrading phases interrupted by instantaneous firing, unless the repressor level becomes sufficiently high to prevent any further firing and to maintain the gene level in a vanishing stationary state. 

Prior to investigating these behaviors in more detail, we first make sure that the dynamics is globally well-posed. As the next statement shows, this is granted by assuming that the evolution begins with a degrading phase in every cell. 
An element $x\in [0,1]^N$ is said to be {\bf admissible} if $W x_i > \eta$ for all $i \in \{1,\cdots,N\}$. (NB: any $x\in [\eta,1]^N$ is admissible.)
\begin{lem}
For any admissible $x\in [0,1]^N$, equation \eqref{eq:diffEQ} has a unique global solution such that $x(0)=x$.
\end{lem}
\begin{proof}
Local existence is a direct consequence of the admissibility condition. Moreover, we have $x_i(t)=(x_i-t)^+$ for all $i$, provided that $t\geq 0$ is sufficiently small. In fact, this expression holds up until a firing occurs.

In addition, for every solution of \eqref{eq:diffEQ}, the function $t\mapsto Wx_i(t)$ is left continuous in every cell; hence we must have $Wx_i(t)\geq \eta$ for all $(i,t)$ (see Lemma 3.1 in \cite{BF14}). Accordingly, the first time $t_1x$ when a firing occurs, {\sl viz.}\ the first {\bf firing time}, is given by 
\[
t_1x=\inf\{s > 0 : W x_i(s) = \eta \text{ for some } i \in \{1,\cdots,N\}\}.
\] 
Clearly, we have $t_1x<+\infty$ (and the infimum here is actually a minimum). 

Let the {\bf firing map} $F$ be defined on admissible vectors $x\in [0,1]^N$ by $Fx=x(t_1x+0)$. The assumption $w_{ii}>\eta$ implies that $Fx$ is also admissible. Hence the second firing time $t_2x=t_1\circ Fx$ is also well-defined and we have $x_i(t)=(Fx_i-t+t_1x)^+$ for $t\in (t_1x,t_2x]$. By induction, one obtains an infinite sequence $\{t_kx\}_{k\in\N}$ of firing times and a unique well-defined solution on every interval $(t_kx,t_{k+1}x]$. 

To conclude, it remains to show that $\lim_{k \to +\infty} t_kx = +\infty$. 
Assume for the sake of contradiction that $t_{\infty} = \lim_{k \to +\infty} t_kx< +\infty$. By the Pidgeonhole Principle, there exists $i \in \{1,\cdots,N\}$ and a subsequence $\{k_n\}_{n\in\N}$ such that $W x_i(t_{k_n}x) = \eta$ for all $n$ and $\lim_{n \to +\infty} t_{k_n}x=t_\infty$. For this $i $, the expression $x_i(t_{k_{n + 1}}x) = \big(1 - t_{k_{n + 1}}x + t_{k_n}x \big)^+$, together with the characterisation of the firing time $t_{k_{n + 1}}x$, implies the estimate
\[
\eta = W x_i(t_{k_{n + 1}}x) \geq \big(1 - t_{k_{n + 1}}x + t_{k_n}x \big)^+ w_{ii}.
\]
Using that $\lim_{n\to +\infty}t_{k_{n + 1}}x - t_{k_n}x=0$, we conclude that $\eta \geq w_{ii}$, contradicting the original assumption on the self-influencing weights $ w_{ii}$.
\end{proof}

\section{Properties of the firing events}
As indicated above, aymptotic behaviors of expression levels depend on the repetitive nature of firing events and their spatial structure with respect to the population. Space-time firing patterns themselves hinge on intercellular coupling topology and, possibly, also on weight intensities. 

In the case of mean-field coupling, firing patterns are strongly structured: every cell $i$ must fire infinitely often, and after each firing in a single cell $i$, every other cell $j\neq i$ must fire either before or simultaneously with the next firing of $i$ (NB: the first alternative occurs for at least one cell $j$, {\sl viz.}\ full synchronization can never occur in this system \cite{BF14}).  

In other cases of coupling, DF systems may have solutions in which some genes never fire, or eventually stop firing. For instance, for $N=2$ and $1-w_{22}>w_{11}(>\eta)$, the expression levels given by\footnote{For $r\in\R$ the expression $\lfloor r \rfloor$ denotes the floor function of $r$. In particular, $\lfloor r - 0 \rfloor = r - 1$ for $r \in \N$.} 
\[
x_1(t)=1-\left(t-\left\lfloor\frac{t}{t_1}-0\right\rfloor t_1\right)\  \text{and}\ x_2(t)=0,\ \forall t>0,
\]
define a periodic solution of equation \eqref{eq:diffEQ} with period $t_1=1-\frac{\eta}{w_{11}}$ (which attracts every trajectory with initial condition $(x_1,x_2)$ with $x_2<x_1$, see section \ref{S-2OSCIL} below for a complete analysis of $N=2$ systems). 

However, by imposing that the neighbor influence be sufficiently small ({\sl i.e.}\ weak coupling regime), one can make sure that any given site fires infinitely often, for every coupling topology.
\begin{lem}
Assume that $w_{ii}>1-\eta$ for some $i\in \{1,\cdots,N\}$. Then for every solution of equation \eqref{eq:diffEQ} with admissible initial condition $x$, we have $x_i(t)>0$ for all $t\in \R^+$ and $x_i(t_{i,k}+0)=1$ for an infinite sequence $\{t_{i,k}\}_{k\in\N}$ of reset times such that 
\[
t_{i,1}\leq \frac{1-\eta}{w_{ii}}\ \text{and}\ 0<t_{i,k+1}-t_{i,k}\leq \frac{1-\eta}{w_{ii}}, \forall k\in\N.
\] 
\label{NOVANISH}
\end{lem} 
\begin{proof} 
Using that the repressor field can never be smaller than $\eta$ in every solution \cite{BF14}, together with the weight normalisation, we have
\[
\eta \leq Wx_i(t) \leq w_{ii} x_i(t) + 1 - w_{ii},
\]
which implies $x_i(t) \geq  \frac{\eta  - (1-w_{ii})}{w_{ii}}>0$ as desired. Moreover, this inequality imposes that the length of every degrading phase cannot exceed $1-\frac{\eta  - (1-w_{ii})}{w_{ii}}=\frac{1-\eta}{w_{ii}}$; {\sl i.e.}\ any two consecutive resets of $x_i$ have to take place within a time interval of length $\frac{1-\eta}{w_{ii}}$.  
\end{proof}

Depending on parameters, we also suspect that two consecutive firings can happen in a cell without the firing of any other gene in the interim. However, similarly to mean-field coupling, such events never occur, provided that all weights from influencing cells are sufficiently small (and weight conservation is assumed). 
\begin{lem}
Assume that $W$ is doubly-stochastic and suppose that $\max_{j\neq i}w_{ij}< \frac{1}{N}$ for some $i\in \{1,\cdots,N\}$. Then, after any firing in cell $i$ alone, there must be a firing in another cell $j\neq i$, before $i$ fires again.
\end{lem}
\begin{proof}
We prove that, given any vector $x\in [0,1]^N$ whose (unique) maximal coordinate is $x_i$, there must be $j\neq i$ such that $Wx_j<Wx_i$. By contradiction, assume that $Wx_i\leq Wx_j$ for all $j\neq i$. By summing these inequalities over $j$ and using the definition of the repressor field, one gets
\[
(N-1)\sum_{k=1}^Nw_{ik}x_k\leq \sum_{k=1}^N\sum_{j\neq i}w_{jk}x_k=\sum_{k=1}^N(1-w_{ik})x_k,
\]
(where the equality follows from the assumption $\sum_jw_{jk}=1$ for all $k$) which in turn yields $0\leq \sum_{k=1}^N(1-Nw_{ik})x_k$. 

However, letting $m=\max_{j\neq i}x_j$ and using the assumption $\max_{j\neq i}w_{ij}<\frac{1}{N}$, we get 
\[
\left| \sum_{k\neq i}(1-Nw_{ik})x_k\right|\leq \sum_{k\neq i}(1-Nw_{ik})m=(Nw_{ii}-1)m,
\]
and we must have $Nw_{ii}-1>0$. Since $m<x_i$, we conclude that $\left| \sum_{k\neq i}(1-Nw_{ik})x_k\right|<(Nw_{ii}-1)x_i$, hence the inequality $0\leq \sum_{k=1}^N(1-Nw_{ik})x_k$ is impossible. 
\end{proof}

We aim to relate asymptotic properties of trajectories to the ordering in which cells fire. Towards that goal, for simplicity we shall focus on trajectories for which a unique cell resets at each firing. Indeed, when all weights are distinct, simultaneous resets in several cells are believed to be non-generic (measure zero) events in the long term dynamics, except for trajectories falling into synchrony subspaces when the coupling possesses a symmetry \cite{SGP03}, as in the mean-field case. However, uniqueness in this case can always be recovered by passing to the quotient network. Therefore, the single cell reset limitation is barely restrictive in the analysis of asymptotic properties.

Given such a trajectory and $k\in\N$, let $i_k\in\{1,\cdots,N\}$ be the reset cell at $k$th firing. We call $\{i_k\}_{k\in\N}$ a {\bf firing sequence}. In this context, the notion of a trajectory compatible with a given firing sequence is obvious. 
A firing sequence in which each segment $\{i_k\}_{k=\ell N+1}^{(\ell+1)N}$ ($\ell\in\Z^+$) is a permutation of $\{1,\cdots,N\}$ is said to be {\bf exhaustive}. A trajectory in which every expression level vanishing at some time is reset at the next firing above is said to be {\bf non-degenerate}. 

In the case of the mean-field coupling, every firing sequence 
must be exhaustive and every trajectory must be non-degenerate (unless genes cluster into groups with identical levels, as previously mentioned). 

For different coupling topologies, this may not always be the case. For instance, for $N=3$, $\eta<\frac19$, 
\[
w_{12}=w_{21}=3\eta,\ \text{and}\ w_{11}=w_{22}=w_{32}=w_{31}=4\eta,
\]
the expression levels given by 
\[
x_1(t)=1-\left(t-\left\lfloor\frac{t}{t_1}-0\right\rfloor t_1\right),\ x_2(t)=1-t_1-\left(t-\left\lfloor\frac{t}{t_1}-0\right\rfloor t_1\right)\ \text{and}\ x_3(t)=0,\ \forall t>0,
\]
where $t_1=\frac23$, define a degenerate periodic solution of equation \eqref{eq:diffEQ} in which genes fire after their level has reached 0, and cell 3 never fires. More precisely, we have $x_i(t)=0$ and $x_{3-i}(t)=1-t_1$ when $Wx_i(t)=\eta$ for $i=1,2$.

For simplicity, in the rest of this paper, we shall only consider non-degenerate solutions with an (eventually) exhaustive firing sequence. 

\section{Analysis of periodic orbits}
In this section, we study {\bf periodic} orbits that return to their initial state after every cell has fired exactly once. The main statement below claims that, when the graph generated by the adjacency matrix $W^T$ is fully connected, non-degenerate periodic orbits are entirely determined by their firing sequences.
\begin{prop}
Assume that $W$ is doubly-stochastic and irreducible. Then, given any periodic exhaustive firing sequence, either no compatible periodic non-degenerate orbit exists, or such a trajectory is unique.  
\label{UNIPER}
\end{prop}
\begin{proof}
Let $\iota=(i_k)_{k=1}^N$ be the initial (fundamental) word of the firing sequence and let $(t_k)_{k=1}^N$ be the corresponding firing times. Consider a non-degenerate (single cell reset) trajectory with firing segment $\iota$ for $t\in [0,t_N]$. Its coordinates $(x_{i_k}(t_N+0))_{k=1}^N$ after the $N$th firing, and hence the initial coordinates $(x_{i_k})_{k=1}^N$ if we also assume this trajectory is periodic, must be given by 
\begin{equation}
x_{i_k}=x_{i_k}(t_N+0)=1-t_N+t_k,\ \forall k\in\{1,\cdots,N\}.
\label{eq:periodicity}
\end{equation}
We aim to show that there is at most one such trajectory, {\sl i.e.}\ that the vector $(t_k)$ is uniquely determined by $\iota$. Notice that monotonicity of the firing times implies $x_{i_k}<x_{i_{k+1}}$, and also $x_{i_1}>0$, because the trajectory is assumed non-degenerate and $x_{i_N}=x_{i_N}(t_N+0)=1$. (If we had $x_{i_1}=0$ then $i_1$ would fire jointly with $i_N$ and this would contradict non-degeneracy.)

Expression \eqref{eq:periodicity} implies that the difference $x_{i_k}-t_k=1-t_N$ does not depend $k$. To proceed, we separate the case $1-t_N>0$, for which no expression level ever reaches 0, from the case $1-t_N\leq 0$ when the reset cell has vanishing expression level immediately prior to its firing. (NB: this cell must be the only one with vanishing level, thanks again to the non-degeneracy assumption.) 

Assume the first case $1-t_N>0$. From expression \eqref{eq:periodicity} and equation \eqref{eq:diffEQ}, one obtains the following gene expression levels immediately before the $k$th firing:
\[
x_{i_j}(t_k)=\left\{\begin{array}{ccl}
1-t_k+t_j&\text{if}&j<k \, ,\\
1-t_N+t_j-t_k&\text{if}&j\geq k \, ,
\end{array}\right.
\]
hence for the repressor field,
\[
W x_{i_k}(t_k)=1-t_N\sum_{j=k}^Nw_{i_ki_j}+Wx_{i_k}-x_{i_k} \, .
\]
The assumption $\sum_i w_{ij}=1$ for all $j$ implies $\sum_k Wx_{i_k}=\sum_k x_{i_k}$. Together with the definition $W x_{i_k}(t_k)=\eta$ of the firing times, this yields the following unique characterization of the firing time $t_N$:
\[
t_N=\frac{N(1-\eta)}{\sum_{k=1}^N\sum_{j=k}^Nw_{i_ki_j}} \, .
\]
In particular, the inequality $1-t_N>0$ holds iff
\begin{equation}
N(1-\eta)<\sum_{k=1}^N\sum_{j=k}^Nw_{i_ki_j} \, .
\label{eq:signCond}
\end{equation}
We proceed similarly in the case $1-t_N\leq 0$. Using the non-degeneracy assumption, the expression levels immediately before the $k$th firing are now given by
\[
x_{i_j}(t_k)=\left\{\begin{array}{ccl}
1-t_k+t_j&\text{if}&j<k \, , \\
0&\text{if}&j=k \, ,\\
1-t_N+t_j-t_k&\text{if}&j> k \, ,
\end{array}\right.
\]
from which we obtain
\[
W x_{i_k}(t_k)=1-w_{i_ki_k}-t_N\sum_{j=k+1}^Nw_{i_ki_j}+Wx_{i_k}-x_{i_k} \, ,
\]
and then
\[
t_N=\frac{N(1-\eta)-\sum_{k=1}^Nw_{i_ki_k}}{\sum_{k=1}^N\sum_{j=k+1}^Nw_{i_ki_j}} \, .
\]
In particular, the existence condition $1-t_N\leq 0$ amounts to 
\[
N(1-\eta)\geq \sum_{k=1}^N\sum_{j=k}^Nw_{i_ki_j} \, ,
\]
which is complementary to the existence condition \eqref{eq:signCond} in the previous case. 

With $t_N$ uniquely specified in terms of the weights $w_{ij}$ and the firing sequence $\iota$, we have that the firing time equations 
\[
W x_{i_k}(t_k)=\eta,\ k\in\{1,\cdots,N\}
\]
can be rewritten in vector form as $(\text{Id}-W)x=u$, where $u$ only depends on $w_{ij}$ and $\iota$ and belongs to the space 
\[
\Sigma=\left\{x\in\R^N\ : \ \sum_i x_i=0\right\} \, ,
\]
orthogonal to the eigenvector $(1,\cdots,1)$. Writing $x=x_\parallel+x_\perp$, where $x_\parallel=c(1,\cdots,1)^T$ for some $c\in\R$ and $x_\perp\in\Sigma$, the firing time equations become $(\text{Id}-W)x_\perp=u$. Moreover, given any $x_\perp$, the constant $c$ is determined by using the normalization $x_{i_N}=1$, {\sl i.e.}\ $c=1-(x_\perp)_{i_N}$.

Now, using the irreducibility of $W$, the Perron-Frobenius Theorem implies that the leading eigenvalue 1 of $W$ is simple, with eigenvector  $(1,\cdots,1)^T$. Therefore, the equation $(\text{Id}-W)x_\perp=u$ has at most a single solution in $\Sigma$. 

At this stage, the proposition is proved. However, we complement the proof by showing that the equation $(\text{Id}-W)x_\perp=u$ has a solution in $\Sigma$.
Notice that the stochastic matrix $\left(W^N\right)^T$ is scrambling, \emph{i.e.} for each pair of rows of $\left( W^N \right)^T$, there is a column on which both rows have strictly positive entries. Indeed, the irreducibility of $W$  implies that, for any pair $i,j$, there is a word $(i_\ell)_{\ell=1}^L$ with $L\leq N-2$ such that 
\[
w_{ii_1}w_{i_1i_2}\cdots w_{i_{L-1}i_L}w_{i_Lj}>0 \, .
\]
Using the assumption $w_{ii}>\eta$ for all $i$, we can assume that all such chains have length $N$. This shows that any two columns of $W^N$ have at least one positive element in a coincident row.
That $\left(W^N\right)^T$ is stochastic and scrambling implies that the restriction $W^N|_{\Sigma}$ is a contraction for the $\ell^1$-norm on $\R^N$ (Exercise 2.27 in \cite{S73}). Moreover, $u\in\Sigma$ implies that the sum $\sum_{\ell=0}^{N-1}W^\ell u\in\Sigma$. Therefore, the series 
\[
\sum_{k=0}^{+\infty}W^{Nk}\sum_{\ell=0}^{N-1}W^\ell u\in\Sigma
\]
converges in $\ell^1(\R^N)$ and solves the equation $(\text{Id}-W)x_\perp=u$. 
\end{proof}
For future purposes, we denote by $x_\iota$ the solution of the equation $(\text{Id}-W)x=u$ in 
\[
{\cal U}_{\iota}=\left\{x=(x_{i_k})_{k=1}^N\ :\ 0<x_{i_1}<x_{i_2}<\cdots<x_{i_{N-1}}<x_{i_N}=1\right\}.
\]

In addition to proving uniqueness, the analysis of the firing time equations can also yield information about the existence of non-degenerate periodic orbits ({\sl viz.}\ whether or not the trajectory issued from $x_\iota$ is actually periodic with firing pattern $\iota$), depending on coupling topology. For instance, in the case of weighted global coupling
\[
Wx_i=(1-\epsilon)x_i+\epsilon\sum_{j=1}^N\alpha_j x_j,\ \forall i\in\{1,\cdots,N\},
\]
where $\alpha_j\geq 0$ and $\sum_{j}\alpha_j=1$, a periodic orbit can exist for every exhaustive firing sequence $\iota=\{i_k\}_{k=1}^{N}$, and it does iff $\epsilon<\epsilon_\iota$ where the coupling threshold $w_\iota$ is known explicitly \cite{BF14}. (NB: Non-uniform weights $\alpha_i\neq \frac{1}{N}$ correspond to effective repressor fields associated with clustered configurations of the mean field coupling.) 

For more general coupling topologies, existence may depend on the firing sequence (in addition to coupling intensity), as illustrated with the following example. 
\begin{ex} ({\sl Nearest neighbor couplings on periodic chains.}) Let $W = (w_{ij})_{ij}$ be given by 
\[
w_{ij}=(1-w)\delta_{ij} +\frac{w}{2}\left(\delta_{(i-1)j}+\delta_{(i+1)j}\right),\ \forall i,j\in\{1,\cdots,N\} \, ,
\]
where the index 0 is identified with $N$, and $N+1$ is identified with 1.
\begin{itemize}
\item[(a)] In the case where $N=2n$ is even, no non-degenerate periodic exists which is compatible with any firing sequence $\{i_k\}_{k=1}^{2n}$ for which 
\[
\left\{\begin{array}{l}
i_k\ \text{is odd for}\ k\in\{1,\cdots,n\} \, ,\\
i_k\ \text{is even for}\ k\in\{n+1,\cdots,2n\} \, .
\end{array}\right.
\]
\item[(b)] The non-degenerate periodic orbit compatible with the firing sequence $\{k\}_{k=1}^{N}$ ({\sl i.e.}\ $i_k=k$) exists iff $w<\frac{2(N-1)}{N-2}\eta$.
\end{itemize}
\end{ex}
\begin{proof}
(a) From the proof of Proposition \ref{UNIPER}, one has that the vector $u$ in the component equation $(I - W) x_\perp = u$ is given by 
\[
u_{i_k}=1-\eta-t_N\sum_{j=k}^Nw_{i_ki_j} 
\]
in the case $1-t_N>0$, and 
\[
u_{i_k}=1-\eta-w_{i_ki_k}-t_N\sum_{j=k+1}^Nw_{i_ki_j}
\]
in the complementary case. For such coupling, the assumption on the firing sequence implies that, in both cases, the components $u_{i_k}$ do not depend on $k\in \{1,\cdots n\}$ (resp.\ on $k\in \{n+1,\cdots 2n\}$); that is, we have $u_{i+2}=u_i$ for all $i$. Hence the vector $x_\perp$, and thus the periodic configuration $x_\iota$ itself, must satisfy the same property, {\sl viz.}\ $(x_\iota)_{i+2}=(x_\iota)_i$. As a consequence, all cells with even (resp.\ odd) indexes must fire together. This property violates the single cell reset assumption.

\noindent
(b) Together with the choice of the firing sequence, that the repressor field commutes with cyclic permutations of coordinates implies that we must have $x_i(t_1+0)=x_{i+1\text{mod} N}(0)$ for all $i$; this implies
\[
(x_\iota)_i=1-(N-i)t_1,\ \forall i\in\{1,\cdots,N\}
\]
and $t_N=Nt_1$. Using again the symmetry of the weights $W$, it suffices to check existence on the time interval $[0,t_1]$. 

Strict monotonicity of the gene levels implies strict monotonicity of the repressor levels. Hence, the firing time equation $Wx_1(t_1)=\eta$ implies $Wx_i(t)>\eta$ for all $i>1$ and $t\in [0,t_1]$  as desired, as well as $Wx_1(t)>\eta$ for $t\in [0,t_1)$.

It remains to check that the gene levels are positive/zero as prescribed in each case $1-t_N>0$ or $1-t_N\leq 0$. In the first case, which explicit calculations show occurs iff $w<2\eta$, we obviously have $x_1(t)>0$, and then $x_i(t) > 0$ for all $i$ by monotonicity, for all $t\in [0,t_1]$. For such $w$, the periodic orbit exists and its coordinates never reach 0.

In the complementary case $w\geq 2\eta$, we know that the firing cell fires from 0, {\sl i.e.}\ $x_1(t_1)=0$. Therefore, we must ensure that all other gene levels remain positive, which by monotonicity is implied by $x_2(t_1)>0$. Explicit calculations show that this is equivalent to $w<\frac{2(N-1)}{N-2}\eta$. The conclusion now follows the uniqueness of periodic orbits as in Proposition \ref{UNIPER} (and that $2\eta<\frac{2(N-1)}{N-2}\eta$).
\end{proof}

\section{Asymptotic periodicity of trajectories with exhaustive firing sequences}
We are now in position to formulate the main result of this paper. 
\begin{thm}
Assume that $W$ is doubly stochastic, irreducible, and that there exists a pair $\mathpzc{i},\mathpzc{j}$ of cells such that $w_{\mathpzc{j}\mathpzc{i}}w_{\mathpzc{i}\mathpzc{j}}>0$. Then, for any trajectory $t\mapsto x(t)$ which is compatible with a given $N$-periodic exhaustive firing sequence with initial word $\iota=\{i_k\}_{k=1}^N$, we have
\[
\lim_{k\to+\infty} x(t_{kN}+0)=x_\iota,
\]
where $t_\ell$ ($\ell\in\N$) is the $\ell$th firing time of the trajectory (and as before, $x_\iota\in {\cal U}_\iota$ is the initial condition of the periodic orbit associated with $\iota$.)
\label{ASYMPTHM}
\end{thm}
In addition, by combining the proof below with continuity arguments, one can show that when the periodic orbit associated with $\iota$ exists, for every initial condition in ${\cal U}_\iota$ sufficiently close to $x_\iota$, the firing sequence of the subsequent trajectory is $N$-periodic with initial word $\iota$; this implies that the periodic orbit $x_\iota$ is locally asymptotically stable. 

We note here that the assumption $w_{\mathpzc{i}\mathpzc{j}}w_{\mathpzc{j}\mathpzc{i}}>0$ ensures that the statement of Theorem \ref{ASYMPTHM} holds for any exhaustive firing sequence $\iota$, regardless of the coupling strength (and in particular, whether or not sites fire after reaching zero along the orbit $x(t)$). This assumption can be relaxed if one assumes instead a weak coupling regime, {\sl i.e.}\ $w_{ii}>1-\eta$ for all $i\in\{1,\cdots,N\}$; see Remark \ref{WEAKCOUP} after the proof. 
\begin{proof}
Given $x(0)=x\in {\cal U}_\iota$, let $F_{i_1}x=x(t_1+0)$ be the reset configuration after the first firing. Assuming that the trajectory $x(t)$ is non-degenerate, we have $(F_{i_1}x)_i>0$ for all $i$, and so the map $F_{i_1}$ itself is defined by 
\[
(F_{i_1}x)_i=\left\{\begin{array}{ccl}
x_i-t_1&\text{if}&i\neq i_1 \, ,  \\
1&\text{if}&i=i_1 \, .
\end{array}\right.
\]
The reset maps $F_{i_k}$ for $k\in\{2,\cdots,N\}$ are defined similarly. 
In order to prove the Theorem, we consider the return map 
\[
F_\iota^N=F_{i_N}\circ \cdots\circ F_{i_2}\circ F_{i_1},
\]
and its orbits that never leave the open set ${\cal U}_\iota$. The map $F_\iota^N$ is continuous and piecewise affine. We are going to prove that a sufficiently high iterate $(F_\iota^N)^k$ is a contraction (for the $\ell^1$ norm). It follows that every orbit in ${\cal U}_\iota$ must approach the unique fixed point $x_\iota$ of $F_\iota^N$, which must be located in the closure $\overline{{\cal U}_\iota}$. Theorem \ref{ASYMPTHM} immediately follows. (NB: The periodic orbit associated with $\iota$ exists iff $x\in {\cal U}_\iota$).

In order to prove the promised contraction, we first introduce a change of variables. Let $\Delta=\text{Id}-W$ denote the graph Laplacian operator on $\R^N$. The expression of the firing time $t_k$, and hence the expression of the image $F_{i_k}x$, depends on whether the level $x_{i_k}$ is reset from a positive value, or from 0. Explicit calculations yield 
\[
x_{i_k}-t_k=\left\{\begin{array}{ccl}
\eta-\Delta x_k&\text{if}&x_{i_k}-t_k\geq 0 \, , \\
\frac{\eta-\Delta x_k}{1-w_{kk}}&\text{if}&x_{i_k}-t_k\leq 0 \, ,
\end{array}\right.
\]
recalling that since $W$ was assumed irreducible, we have $w_{kk} < 1$ for all $k \in \{1,\cdots, N\}$. This expression implies the following commutation relation 
\[
\Delta \circ F_{i_k}=G_{i_k}\circ \Delta \, ,
\]
where the maps $G_{i}$ ($i\in\{1,\cdots ,N\}$) are also continuous and piecewise affine, with linear parts $L_{i,\pm}$ depending on whether $x_{i}$ is reset from a positive value ($+$ sign) or from 0 ($-$ sign). The linear parts $L_{i, \pm}$ are given by
\[
L_{i,+}x_j=\left\{\begin{array}{ccl}
x_j+w_{ji}x_i&\text{if}&j\neq i \, , \\
w_{ii}x_i&\text{if}&j=i \, ,
\end{array}\right.
\quad\text{and}\quad
L_{i,-}x_j=\left\{\begin{array}{ccl}
x_j+\frac{w_{ji}}{1-w_{ii}}x_i&\text{if}&j\neq i \, , \\
0&\text{if}&j=i \, ,
\end{array}\right.
\]
for all $j\in\{1,\cdots,N\}$. Note that the $L_{i, \pm}$'s are column-stochastic by the double stochasticity of $W$.

The change of variable $x \mapsto \Delta x$ has the following properties:
\begin{itemize}
\item[$\bullet$] $\Delta$ is injective on ${\cal U}_\iota$, for every permutation $\iota$ of $\{1,\cdots ,N\}$ (see end of proof of Proposition \ref{UNIPER} above),
\item[$\bullet$] $\Delta(\R^N)\subset \Sigma=\left\{x\in\R^N\ : \ \sum_i x_i=0\right\}$.
\end{itemize}
Consequently, we have 
\[
(F_\iota^N)^k=\Delta^{-1}\circ (G_\iota^N)^k\circ \Delta,\ \forall k\in\N,
\]
where $G_\iota^N=G_{i_N}\circ \cdots\circ G_{i_2}\circ G_{i_1}$, and we regard $\Delta^{-1}$ as a map $\Sigma \to \mathcal U_{\iota}$. The contraction of $(F_\iota^N)^k$ for $k$ large enough follows from the following statement, together with the fact that $\|\Delta^{-1}\|_1\|\Delta\|_1<+\infty$.
\begin{lem}\label{lem:contract}
There exist $k\in\N$ and $\gamma<1$ such that 
\[
\|(G_\iota^N)^kx-(G_\iota^N)^ky\|_1\leq \gamma \|x-y\|_1,\ \forall x,y\in \Sigma.
\]
\end{lem}
\noindent
{\em Proof of the Lemma.} The map $G_\iota^N$ is continuous and piecewise affine with $2^N$ pieces, each corresponding to a choice of sign $\{+, -\}$ for the firing event at each site in $\{1,\cdots,N\}$. Labeling each piece by a symbolic word $s = (s_k)_{k =1}^N \in \{+,-\}^N$, we write
\[
L_{\iota, s} = L_{i_N, s_N} \circ \cdots \circ L_{i_1, s_1} \, .
\]
Because the firing map is continuous across these piecewise domains (see Proposition 5.4 in \cite{BF14}), to prove the statement of Lemma \ref{lem:contract} it will suffice to show the existence of $k$ for which all matrices $(L_{\iota,s})^k$ are contractions on $\Sigma$ in the $\ell^1$-norm. To this end (similarly to the end of the proof of Proposition \ref{UNIPER}), it is enough to show that the row-stochastic matrices $\left((L_{\iota,s})^k\right)^T$ are all scrambling. However, for every $i\in\{1,\cdots ,N\}$, the matrix entries $(L_{i,\pm})_{nm}$ satisfy
\[
(L_{i,-})_{nm}>0\ \Longrightarrow (L_{i,+})_{nm}>0 \, .
\]
Therefore, one only has to prove that $(L^k)^T$ is scrambling for $k$ sufficiently large, where $L=L_{\iota,(-)^N}$.

Up to a relabeling of the cells $\mathpzc{i},\mathpzc{j}$, we can assume that $\mathpzc{i}$ appears before $\mathpzc{j}$ in the word $\iota$, {\sl viz.}\ we have $\mathpzc{i}=i_{k_1}$ and $\mathpzc{j}=i_{k_2}$ with $k_1<k_2$. We are going to prove the existence of $k\in\N$ such that 
\begin{equation}
(L^k)_{\mathpzc{i}j_1}>0,\ \forall j_1\in\{1,\cdots ,N\} \, ,
\label{POSITIVITY}
\end{equation}
which is merely a restatement of the scrambling property for $(L^k)^T$. To show \eqref{POSITIVITY}, we use the following expressions for the entries of the composing matrices $L_{i,-}$:
\begin{equation}
(L_{i,-})_{jk}=\frac{w_{ji}}{1-w_{ii}}\delta_{ik} \, , \ \forall k\neq j\in\{1,\cdots ,N\}\quad\text{and}\quad 
(L_{i,-})_{jj}=1-\delta_{ij} \, ,\ \forall j\in\{1,\cdots ,N\} \, .
\label{CRUCIAL}
\end{equation}
By irreducibility of $W$, given an arbitrary $j_1\in\{1,\cdots,N\}$, let $(j_k)_{k=1}^L$ (where $L\leq N-1$) be the shortest word such that 
\[
w_{\mathpzc{i}j_L}w_{j_Lj_{L-1}}\cdots w_{j_2j_1}>0 \, .
\]
Let $k_1\in\{1,\cdots,N\}$ be such that $i_{k_1}=j_1$. By \eqref{CRUCIAL} and since $1-w_{i_{k_1}i_{k_1}}<1$, we have 
\[
(L_{i_{k_1},-}\cdots L_{i_1,-})_{j_2j_1}\geq w_{j_2j_1} \, .
\]
(This is a consequence of the relations $(L_{i_{k_1-1},-}\cdots L_{i_1,-})_{j_1j_1}=1$ and $(L_{i_{k_1},-}\cdots L_{i_1,-})_{j_2j_1}=(L_{j_1,-})_{j_2j_1}$.) Let now $k_2$ be such that $i_{k_2}=j_2$ and consider separately the cases $k_1<k_2$ and $k_2<k_1$. In the first case, we have 
\[
(L_{i_{k_2},-}\cdots L_{i_{k_1+1},-})_{j_3j_2}=(L_{i_{k_2},-})_{j_3j_2}\geq w_{j_3j_2} \, ,
\]
and then 
\[
(L_{i_{k_2},-}\cdots L_{i_1,-})_{j_3j_1}\geq w_{j_3j_2}w_{j_2j_1} \, .
\]
In the second case, we have 
\[
(L_{i_{k_N},-}\cdots L_{i_{k_1+1},-})_{j_3j_2}=0 \, ,
\]
so no positive estimate holds for $L_{j_3j_1}$. However, we certainly have 
\[
L_{j_2j_1}\geq w_{j_2j_1}\quad\text{and}\quad
(L_{i_{k_2},-}\cdots L_{i_1,-})_{j_3j_2}=(L_{i_{k_2},-})_{j_3j_2}\geq w_{j_3j_2} \, ,
\]
hence
\[
(L_{i_{k_2},-}\cdots L_{i_1,-}L)_{j_3j_1}\geq w_{j_3j_2}w_{j_2j_1} \, .
\]
By repeating this process, we obtain that there exists $n\in\{0,\cdots, L-1\}$ for every $(j_k)_{k=1}^L$ such that 
\[
(L_{i_{k_L},-}\cdots L_{i_1,-}L^n)_{\mathpzc{i}j_1}\geq w_{\mathpzc{i}j_L}w_{j_Lj_{L-1}}\cdots w_{j_2j_1}>0 \, .
\]
To conclude, we use the following property
\[
\begin{array}{cl}
(L_{i_N,-}\cdots L_{i_{k_L+1},-})_{\mathpzc{i}\mathpzc{i}}&=1\ \text{if}\ \mathpzc{i}\not\in\{i_{k_L+1},\cdots,i_N\},\\
&\geq w_{\mathpzc{j}\mathpzc{i}}w_{\mathpzc{i}\mathpzc{j}}\ \text{if}\ \mathpzc{i}\in\{i_{k_L+1},\cdots,i_N\},
\end{array}
\]
where the second estimate follows from the fact that $\mathpzc{j}$ appears after $\mathpzc{i}$ in $\iota$. Using $w_{\mathpzc{j}\mathpzc{i}}w_{\mathpzc{i}\mathpzc{j}}>0$, it results that $(L^{n+1})_{\mathpzc{i}j_1}>0$ for every $\{j_k\}_{k=1}^L$.

Finally, this last estimate also implies $L_{\mathpzc{i}\mathpzc{i}} \geq w_{\mathpzc{j}\mathpzc{i}}w_{\mathpzc{i}\mathpzc{j}}>0$. Hence, letting $k=\max_{\{j_k\}_{k=1}^L} n$, we can always multiply by $L^{k-(n+1)}$ to obtain the desired estimate \eqref{POSITIVITY}.\end{proof}
\begin{rmk}
In the weak coupling regime $w_{ii}>1-\eta$, the map $G_\iota^N$ is affine with linear part $L_{\iota,(+)^N}$. Using the properties
\[
(L_{i,+})_{jk}=\delta_{jk}(1-\delta_{ij})+w_{ji}\delta_{ik},
\]
instead of \eqref{CRUCIAL}, one can repeat the proof {\sl mutatis mutantis}, now using the estimate
\[
(L_{i_N,+}\cdots L_{i_{k_L+1},+})_{\mathpzc{i}\mathpzc{i}}\geq w_{\mathpzc{i}\mathpzc{i}}>0,
\]
in the case $\mathpzc{i}\in\{i_{k_L+1},\cdots,i_N\}$, to obtain the conclusion \eqref{POSITIVITY}. 
\label{WEAKCOUP}
\end{rmk}

\section{Low dimensional examples}
In complement to previous results on arbitrary DF systems, we now present examples for which the dynamics is entirely known. Ignoring the case of mean-field coupling which has previously been described for populations of arbitrary size $N$, we focus on low dimensional systems $N=2$ and $N=3$.

\subsection{Two coupled oscillators ($N=2$)}\label{S-2OSCIL}
Letting $w_2=w_{12}$ and $w_1=w_{21}$ for simplicity, we get 
\[
Wx_i=(1-w_{3-i})x_i+w_{3-i}x_{3-i},\ i=1,2,
\]
and we consider separately the cases $w_1+w_2<1$, $w_1+w_2=1$ and $w_1+w_2>1$. 

\noindent
{\sl $\bullet$ Case $w_1+w_2=1$.} We have $Wx_1=Wx_2$ for all $x\in [0,1]^2$; hence cells 1 and 2 must fire simultaneously and evolve in sync after the first firing, as an $N=1$ DF oscillator. 

Of note, this property extends to any population size $N\in\N$ in the trivial case when no weight $w_{ij}$ depends on $i$, {\sl viz.}\ $Wx_i=\sum_{j}w_{j}x_j$ for all $i$. Similarly, full synchrony in the trajectory holds for any coupling $W$ when the initial coordinates $x_i$ do not depend on $i$. For convenience, from here on, we shall assume that not all weights, nor all coordinates, are equal. 

\begin{figure}[h]
\centerline{\includegraphics[scale=0.45]{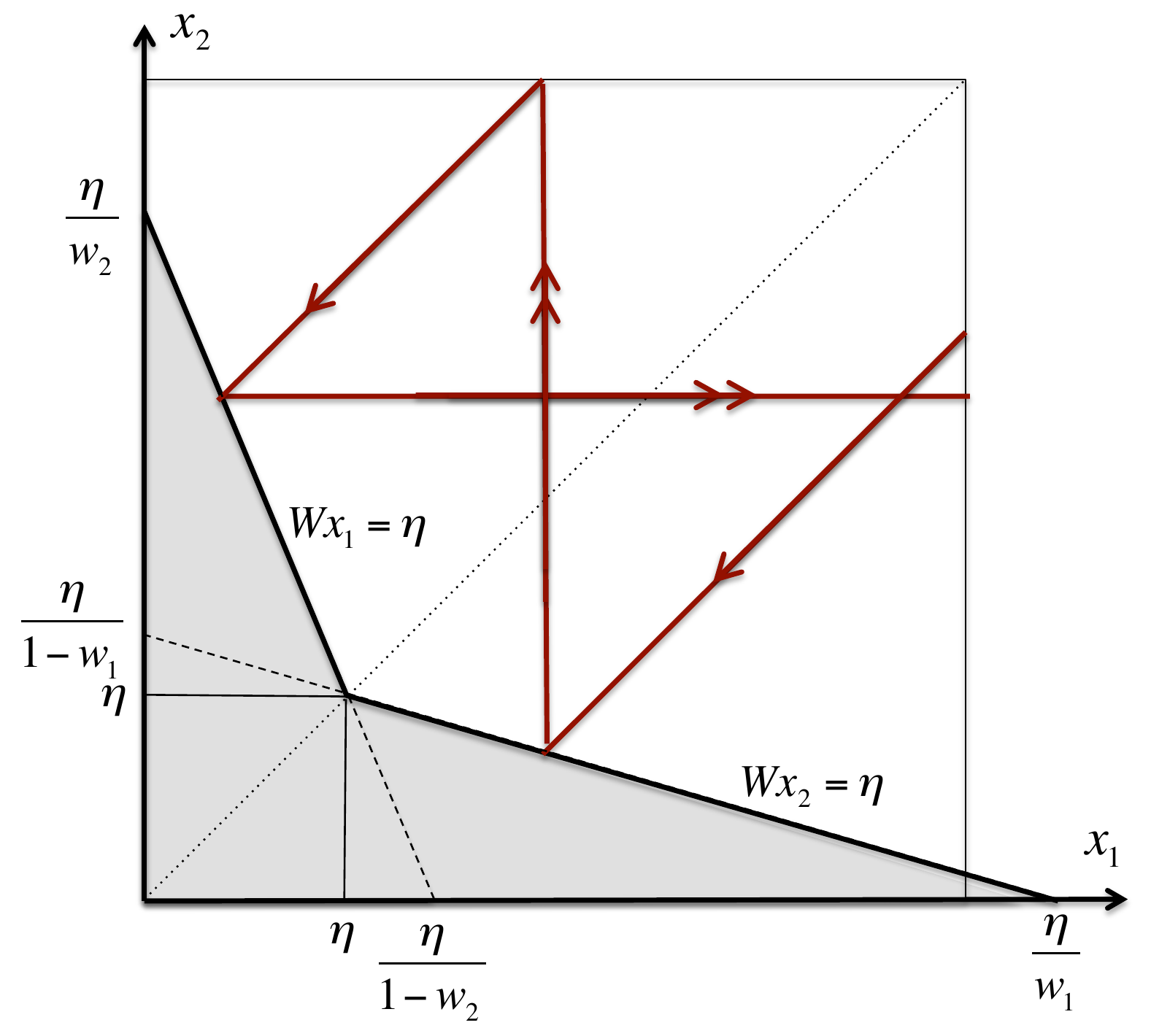}\hspace{0.5cm} \includegraphics[scale=0.45]{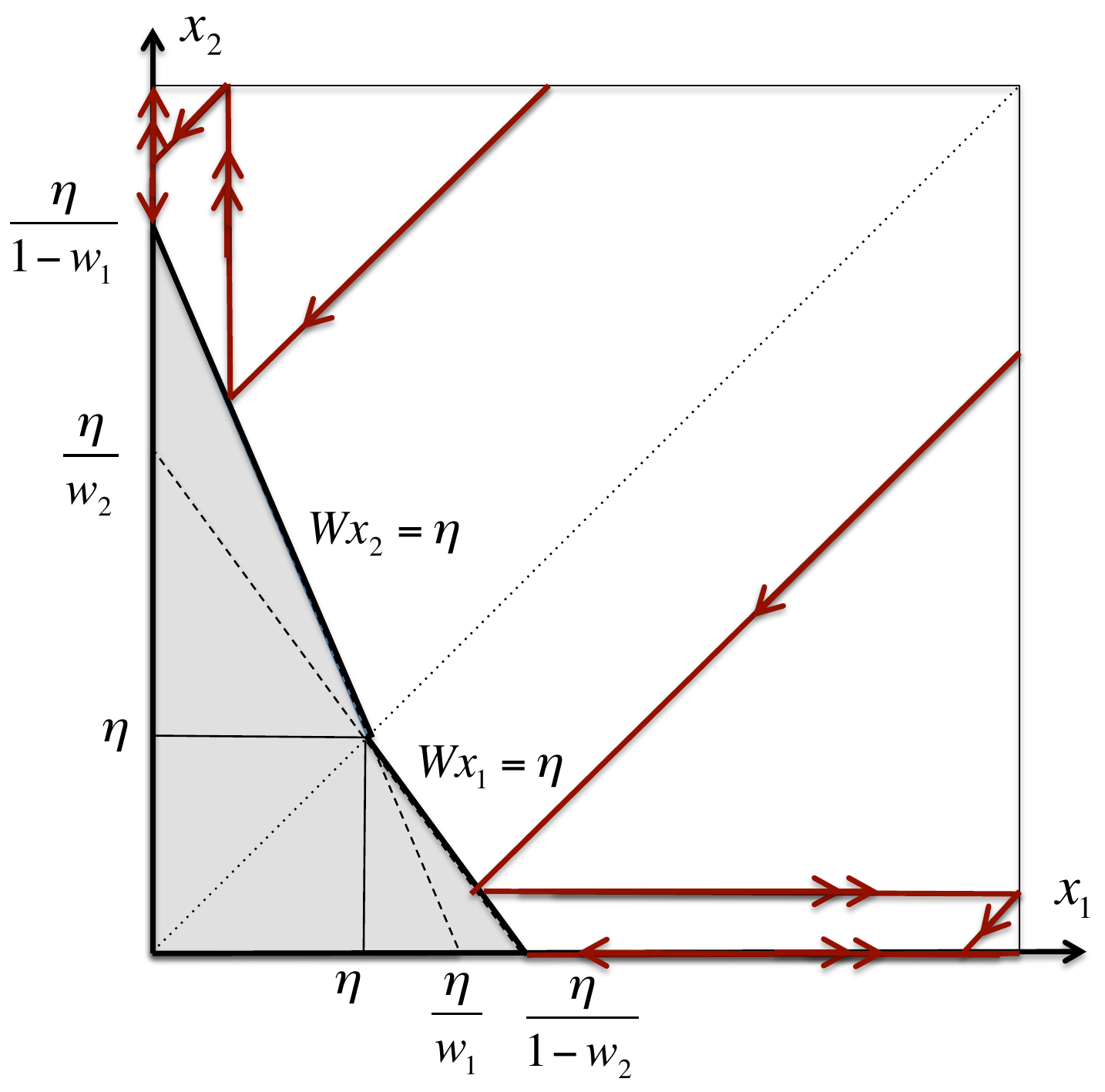}}
\caption{Trajectories in the unit square $[0,1]^2$ for systems of two coupled DF oscillators. Single arrows represent motion at speed 1. Double arrows represent resets at firings. {\sl Left.} $w_1+w_2<1$: every firing sequence must be exhaustive with cells 1 and 2 firing alternatively. {\sl Right.}\ $w_1+w_2>1$: Depending on its initial location with respect to the diagonal $x_1=x_2$, any trajectory reaches in finite time a periodic orbit with either $x_1=0$, or $x_2=0$. NB: The white (un-shaded) region of $[0,1]^2$ corresponds to the set of admissible initial conditions. Recall that we always assume $w_{ii}=1-w_{3-i}>\eta$ for $i=1,2$, so that the dynamics is well-defined. Moreover, $w_i<\eta$ corresponds to the weak coupling case, when expression level $x_i$ cannot reach 0 (see Lemma \ref{NOVANISH}). Conversely, when $w_i\geq \eta$, the level $x_i$ may vanish, depending upon the trajectory.}
\label{TRAJECTN2}
\end{figure}
\noindent
{$\bullet$} {\sl Case $w_1+w_2<1$.} Until it fires, any trajectory initially located in the segment $x_1=1$ lies below the diagonal $x_1=x_2$ of the unit square $[0,1]^2$, and cell 2 must fire first (see left panel in Fig.\ \ref{TRAJECTN2}). Conversely, in any trajectory initially located in the segment $x_2=1$, cell 1 must fire first. It results that every firing sequence (of a trajectory out of the diagonal) must be exhaustive, and cells 1 and 2 must fire alternatively. 

If we also assume that $w_1=w_2>0$ (so that $W$ is doubly stochastic and irreducible since we also assume $1 - w_i > \eta$ from the beginning), using Theorem \ref{ASYMPTHM}, we conclude that every trajectory asymptotically approaches the periodic orbit associated with the firing pattern $\{1,2\}$, for which the corresponding periodic orbit always exists.

\noindent
{$\bullet$} {\sl Case $w_1+w_2>1$.} As opposed to the previous case, for any trajectory initially located in the segment $x_1=1$, only cell 1 can fire; hence $x_2$ eventually reaches and stays at 0. All such trajectories reach in finite time the periodic orbit where cell 1 oscillates alone (whose expression is given before Lemma \ref{NOVANISH}, see right panel in Fig.\ \ref{TRAJECTN2}). Similarly, any trajectory initially at $x_2=1$ reaches a periodic trajectory with $x_1(t)=0$. 

\subsection{Three coupled oscillators ($N=3$)}
For simplicity, we assume that each cell influences all other cells in the same way, $w_{21}=w_{31}(=w_1)$, $w_{12}=w_{32}(=w_2)$ and $w_{13}=w_{23}(=w_3)$, {\sl i.e.}\ we have 
\[
\left\{\begin{array}{l}
Wx_1=(1-w_2-w_3)x_1+w_2x_2+w_3x_3\\
Wx_2=w_1x_1+(1-w_1-w_3)x_2+w_3x_3\\
Wx_3=w_1x_1+w_2x_2+(1-w_2-w_3)x_3
\end{array}\right.
\]
Thanks to this symmetry,  for all pairs $i,j$ of indexes, we have $Wx_i=Wx_j$ in the plane $x_i=x_j$, and then $Wx_1=Wx_2=Wx_3$ along the diagonal of the cube $[0,1]^3$. In order to characterise the dynamics, we need to determine which parts of the planes $Wx_i=\eta$ can be reached under the flow. Similarly to $N=2$, we shall separate the cases $w_1+w_2+w_3<1$ and $w_1+w_2+w_3>1$.

\begin{figure}[h]
\centerline{\includegraphics[scale=0.42]{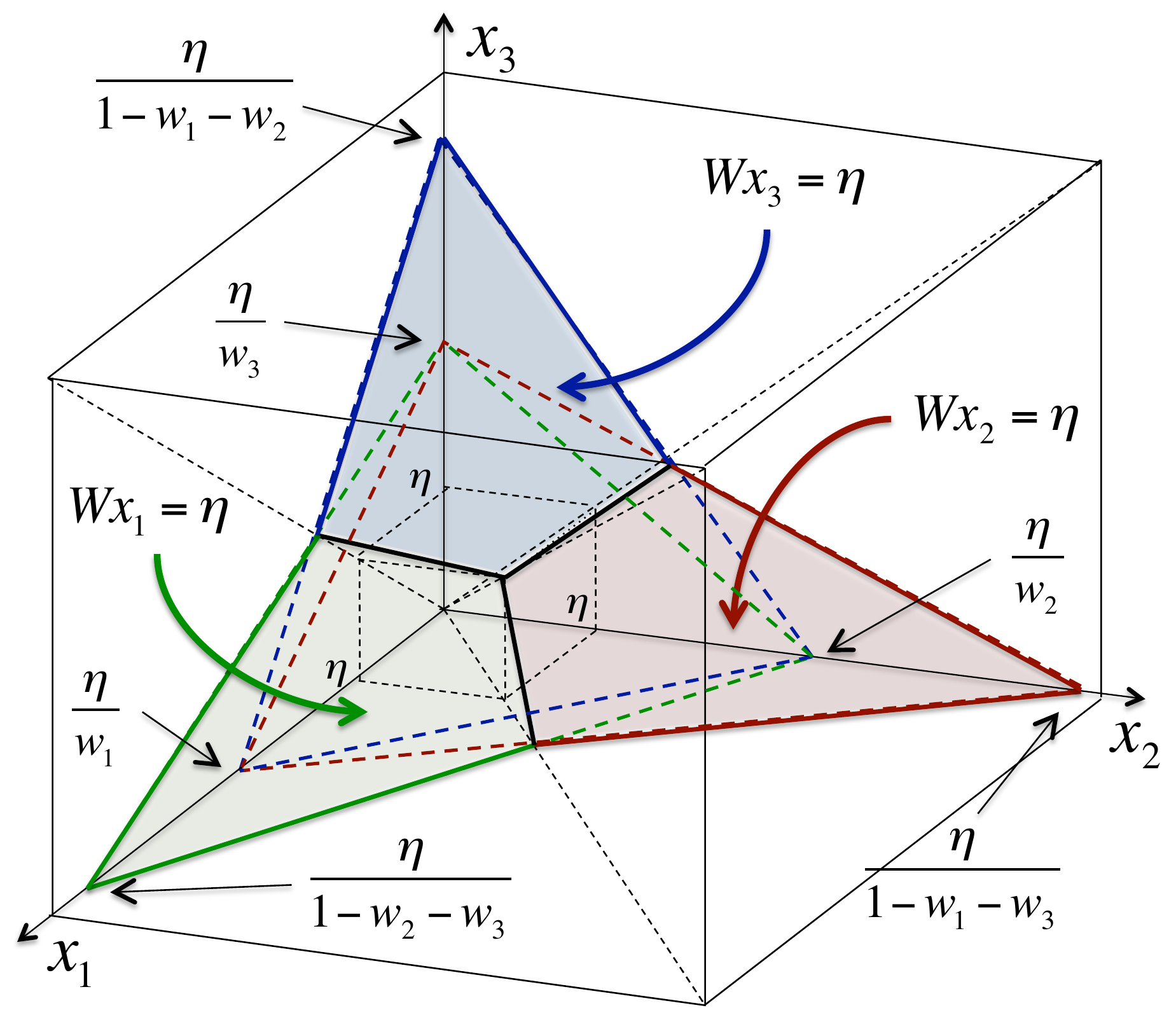}\hspace{0.5cm} \includegraphics[scale=0.42]{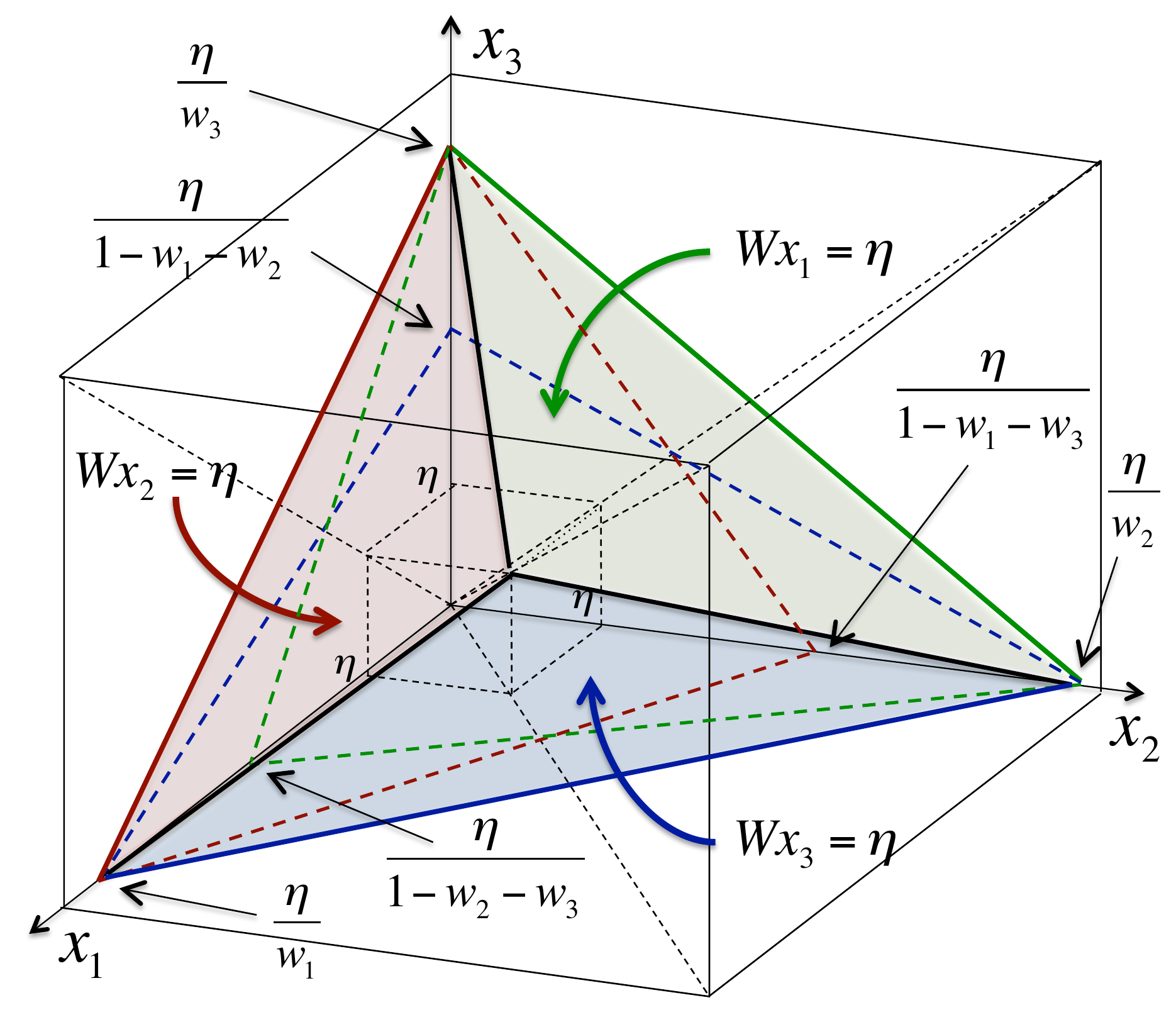}}
\caption{Locations of firing planes $Wx_i=\eta$ in the unit cube $[0,1]^3$ for systems of three coupled DF oscillators with weights satisfying $w_{21}=w_{31}(=w_1)$, $w_{12}=w_{32}(=w_2)$ and $w_{13}=w_{23}(=w_3)$. {\sl Left.}\ $w_1+w_2+w_3>1$: In the pyramidal cone with section delimited by the square $x_1=1$ and with apex at the origin, the plane $Wx_1=\eta$ (green shaded area) lies above the other planes $Wx_2=\eta$ and $Wx_3=\eta$. Since this cone contains all trajectories starting with $x_1=1$, cell 1 must fire first, and by induction, this cell is the only one to fire in this case. By symmetry, similar arguments apply to the cones respectively delimited by $x_2=1$ and $x_3=1$. {\sl Right.}\ $w_1+w_2+w_3>1$: Now, the pyramidal cone $x_1=\max_i\{x_i\}$ decomposes into two regions. In the lower sector $x_2>x_3$, the plane $Wx_3=\eta$ (blue shaded area) lies above the two other planes, while the plane $Wx_2=\eta$ (red shaded area) dominates in the sector $x_2<x_3$. By symmetry, we conclude that the firing sequence of every non-degenerate trajectory must be exhaustive and periodic with pattern $\{1,2,3\}$, $\{1,3,2\}$, or one of their cyclic permutations.}
\label{TRAJECTN3}
\end{figure}

\noindent
{$\bullet$} {\sl Case $w_1+w_2+w_3>1$.} Any trajectory initially located in the square $x_1=1$ remains in the pyramidal cone $x_1=\max_i\{x_i\}$ (delimited by the planes $x_3=0$, $x_1=x_2$, $x_1=x_3$ and $x_2=0$), until it fires. In this sector, the plane $Wx_1=\eta$ lies above the two other planes $Wx_2=\eta$ and $Wx_3=\eta$ (more precisely, the truncated solids respectively delimited by the planes $Wx_2=\eta$ and $Wx_3=\eta$ in this sector, both contain the truncated solid delimited by the plane $Wx_1=\eta$, see Fig.\ \ref{TRAJECTN3}, left panel), hence cell 1 fires first and is reset to the square $x_1=1$. By induction, it results that the trajectory remains in this cone forever. Moreover, both $x_2$ and $x_3$ must eventually vanish and the trajectory reaches in finite time, a periodic orbit with $x_2(t)=x_3(t)=0$. Similar scenarios occur for trajectories initially located in, respectively, the squares $x_2=1$ and $x_3=1$.

Of note, the phase portrait in this case is fully preserved under asymmetric perturbations of parameters, provided that the following conditions hold:
\[
1-w_{31}-w_{32}<w_{13} \wedge w_{23},\ 
1-w_{21}-w_{23}<w_{12} \wedge w_{32},\ \text{and}\ 
1-w_{12}-w_{13}<w_{21} \wedge w_{31}.
\]

\noindent
{$\bullet$} {\sl Case $w_1+w_2+w_3<1$.} In this case, the pyramidal cone $x_1=\max_i\{x_i\}$ decomposes into two regions, according to the sign of $x_2-x_3$ (ignoring the case $x_2=x_3$). In the lower sector $x_2>x_3$, the plane $Wx_3=\eta$ lies above the two other ones (Fig.\ \ref{TRAJECTN3}, right panel), while in the upper sector $x_2<x_3$, the plane $Wx_2=\eta$ dominates. 

Moreover, a trajectory initially in the lower sector is located, after the first firing that resets cell 3, in the sector $x_1>x_2$ of the square $x_3=1$. Therefore, cell 2 must fire second, then cell 1, and then the trajectory is back into the lower section of $x_1=1$. By induction, it results that the trajectory has exhaustive firing sequence with periodic pattern $\{3,2,1\}$. As known from Theorem \ref{ASYMPTHM}, in the doubly stochastic and irreducible case $w_1=w_2=w_3>0$ (which is equivalent to mean-field coupling), it must asymptotically approach the associated periodic trajectory.

Similarly, trajectories starting in the upper sector $x_2<x_3$ of the pyramidal cone $x_1=\max_i\{x_i\}$ have exhaustive firing sequence with periodic pattern $\{2,3,1\}$. Moreover, the fate of trajectories starting form the other cones $x_2=\max_i\{x_i\}$ and $x_3=\max_i\{x_i\}$ can be obtained in the same way, by applying the permutation symmetries. Hence, the dynamics is also fully described in this case.
\medskip

\noindent
{\bf Acknowledgements} 

\noindent
Work supported by CNRS PEPS "Physique Th\'eorique et ses Interfaces".


\begin{thebibliography}{99}
\bibitem{AD-GKMZ08}
A.\ Arenas, A.\ Diaz-Guilera, J.\ Kurths, Y.\ Moreno and C.\ Zhou, \emph{Synchronization in complex networks}, Phys.\ Rep.\  {\bf 469} (2008) 93--153.

\bibitem{BF14} 
A.\ Blumenthal and B.\ Fernandez, \emph{Population dynamics of globally coupled degrade-and-fire oscillators}, preprint available online at {\tt https://hal.archives-ouvertes.fr/hal-00986128}.

\bibitem{BCFV02}
G.\ Boffetta, M.\ Cencini, M.\ Falcioni and A.\ Vulpiani, \emph{Predictability: a way to characterize complexity}, Phys.\ Rep.\  {\bf 356} (2002) 367--474.

\bibitem{B95}
S.\ Bottani, \emph{Pulse-coupled relaxation oscillators: from biological synchronization to self-organized criticality}, Phys.\ Rev.\ Lett.\ {\bf 74}  (1995) 4189.

\bibitem{B06} 
A.N.\ Burkitt, \emph{A review of the integrate-and-fire neuron model: I. homogeneous synaptic input}, Biol.\ Cybern.\ {\bf 95} (2006) 1-19.

\bibitem{DM-PTH10}
T.\ Danino, O.\ Mondragon-Palomina, L.S.\ Tsimring and J.\ Hasty, \emph{A synchronized quorum of genetic clocks} Nature {\bf 463} (2010) 326--330. 

\bibitem{EPG95} 
U.\ Ernst, K.\ Pawelzik and T.\ Geisel, \emph{Synchronization induced by temporal delays in pulse-coupled oscillators}, Phys.\ Rev.\ Lett.\ {\bf 74} (1995) 1570.

\bibitem{FT11}
B.\ Fernandez and L.S.\ Tsimring, \emph{Corepressive interaction and clustering
  of degrade-and-fire oscillators}, Phys. Rev. E \textbf{84} (2011), 051916.

\bibitem{FT13}
B.\ Fernandez and L.S.\ Tsimring, \emph{Typical trajectories of coupled degrade-and-fire oscillators:
  from dispersed populations to massive clustering}, J. Math. Bio. {\bf 68} (2014) 1627--1652.
  
\bibitem{M96} 
R.S.\ MacKay, \emph{Dynamics of networks: Features which persist from the uncoupled limit}, in Stochastic and spatial structures of dynamical systems, Lunel (1996) 81--104.

\bibitem{MHT14}
W.\ Mather, J.\ Hasty and L.S.\ Tsimring, \emph{Synchronization of degrade-and-fire oscillations via a common activator}, Phys.\ Rev.\ Lett.\ {\bf 113} (2014) 128102.

\bibitem{MBHT09}
W.\ Mather, M.R.\ Bennet, J.\ Hasty, and L.S.\ Tsimring, \emph{Delay-induced degrade-and-fire oscillations in small genetic circuits}, Phys. Rev. Lett. \textbf{102} (2009), 068105.

\bibitem{MS90}
R.\ Mirollo and S.H.\ Strogatz, \emph{Synchronization of pulse-coupled biological oscillators}, SIAM J.\ Appl.\ Math.\ {\bf 50} 1645--1662.

\bibitem{M-PDSTH11}
O.\ Mondragon-Palomino, T.\ Danino, J.\ Selimkhanov, L.\ Tsimring and J.\ Hasty, \emph{Entrainment of a population of synthetic genetic oscillators}, Science {\bf 333} (2011) 1315--1319.

\bibitem{SU00}
W.\ Seen and R.\ Urbanczik, \emph{Similar non leaky integrate-and-fire neutrons with instantaneous couplings alwayys synchronise},  SIAM J.\ Appl.\ Math.\ {\bf 61} (2000) 1143--1155.

\bibitem{S73} 
E.\ Seneta, \emph{Non-negative matrices and markov chains}, Springer (1973). 

\bibitem{SGP03}
I.\ Stewart, M.\ Golubitsky and M.\ Pivato, \emph{Symmetry groupoids and patterns of synchrony in coupled cell networks}, SIAM J.\ Appl.\ Dynam.\ Sys.\ {\bf 2} (2003) 609--646

\bibitem{S01}
S.H.\ Strogatz, \emph{Exploring complex networks}, Nature {\bf 410} (2001) 268--276.
\end{thebibliography}
\end{document}